\documentclass[10pt,journal]{IEEEtran}

\usepackage{cite}
\usepackage{amsmath, amssymb, amsfonts, amsthm}
\usepackage{algorithm}
\usepackage[noend]{algorithmic}
\usepackage{graphicx}
\usepackage{textcomp}
\usepackage{xcolor}
\usepackage{bm}
\usepackage{booktabs}
\usepackage{autobreak}
\usepackage{multirow}
\usepackage{authblk}
\usepackage{mathrsfs}
\usepackage{array}
\usepackage{makecell} 
\usepackage{float}
\usepackage{soul}
\usepackage{stmaryrd}
\usepackage{dsfont}
\usepackage{bbm}
\usepackage[hidelinks]{hyperref}
\usepackage{cleveref}
\usepackage{pifont}
\usepackage{enumitem}
\usepackage{flushend}

\newcommand{\Tr}{\operatorname{Tr}}
\newcommand{\Sset}{\mathcal{S}}
\newcommand{\Aset}{\mathcal{A}}
\newcommand{\Oset}{\mathcal{O}}
\newcommand{\Mset}{\mathcal{M}}
\newcommand{\Tset}{\mathcal{T}}
\newcommand{\CVaR}{\mathrm{CVaR}}
\newcommand{\OCE}{\mathrm{OCE}}

\newtheorem{proposition}{Proposition}
\newtheorem{remark}{Remark}
\newtheorem{corollary}{Corollary}

\begin{document}

\title{Risk-Averse Decision Making via Quantum Measurement Design}

\author{Meiyi Zhu, \IEEEmembership{Member, IEEE}, and Osvaldo Simeone, \IEEEmembership{Fellow, IEEE}

\vspace{-2em}

\thanks{The work of M. Zhu and O. Simeone was supported by an Open Fellowship of the EPSRC (EP/W024101/1). The work of O. Simeone was also supported by EPSRC (EP/X011852/1) and ERC (No. 101198347)

Meiyi Zhu is with the Department of Engineering, King's College London, WC2R 2LS, London, U.K. (e-mail: meiyi.1.zhu@kcl.ac.uk).

Osvaldo Simeone is with the Institute for Intelligent Networked Systems, Northeastern University London, E1 8PH London, U.K. (e-mail: o.simeone@northeastern.edu).}
}

\maketitle

\begin{abstract}
Quantum measurements are conventionally optimized to maximize the average of a utility that depends on the true state and on the measurement outcome. However, when the outcome of the measurement is used as an action within a larger decision-making system, the average utility does not capture the risk of poor outcomes. This letter addresses the design of quantum measurements that maximize a risk-averse objective given by the optimized certainty equivalent (OCE), a family of criteria that includes the average utility and the conditional value at risk (CVaR) as special cases. For a piecewise linear gain function, defining the OCE, thus including the CVaR, the problem is shown to reduce to a finite number of semidefinite programs, for which a dual formulation is derived. For the discrimination of two states, a closed-form solution is obtained that takes the form of a Helstrom measurement. Numerical results show that the optimized measurement improves the lower tail of the utility distribution at a moderate cost in average utility.
\end{abstract}

\vspace{-0mm}

\IEEEpeerreviewmaketitle

\vspace{-1mm}
\section{Introduction}
The state of a quantum system cannot be read directly, and information about it can be acquired only by performing a measurement. Moreover, a measurement generally disturbs the state on which it acts, so that repeated observations of the same system are generally not possible \cite{helstrom, holevo,wilde2013quantum,simeone2026}. The measurement must therefore be committed to in advance, and it must be selected so as to extract the information that is most valuable for the task at hand. The design of quantum measurements is accordingly a central problem in quantum metrology and quantum information theory, underlying applications that range from quantum optical communications \cite{krovi2015optimal, krikidis2025}, to quantum sensing \cite{nikoloska2025adaptive, cui2026}, quantum statistical inference \cite{liu2026unifying, cumitini2026}, and quantum decision making \cite{barnett, busemeyer2012quantum}.

Depending on the application, the measurement is designed to maximize a utility that captures the value of the outcome. Notably, in state discrimination the utility is the accuracy with which the state is identified in a known ensemble \cite{helstrom, barnett}, and in quantum sensing the utility quantifies the accuracy of parameter estimation \cite{nikoloska2025adaptive, cui2026}. In all of these formulations, the design goal is specifically the \emph{mean} utility, averaged over the prior on the state and on the randomness of the measurement outcomes. 

With a mean utility objective, the optimization problem often reduces to a semidefinite program (SDP) \cite{eldar2003, nakahira2017}, whose solution is characterized by a set of necessary and sufficient conditions \cite{ykl, holevo}. Closed-form solutions are available in special cases, including, notably, the Helstrom measurement for the discrimination of two states \cite{helstrom}. More broadly, so-called pretty good measurements are known to be near-optimal for general ensembles and exactly optimal under a group symmetry \cite{eldarforney, krovi2015optimal}. Related constructions have recently been unified as solutions of a relative-entropy criterion in \cite{liu2026unifying}.

Quantum systems are increasingly being considered as \emph{action policies} within larger decision-making systems. For example, in quantum reinforcement learning, a parameterized quantum circuit encodes the observation of the environment into a quantum state, and the action applied to the system is obtained by measuring that state, either by mapping the outcome directly to an action or by post-processing the measured observables \cite{jerbi2021, chen2020, skolik2022, zhao2025}. When policies of this type are deployed in safety-critical loops, the average utility is often not a sufficient design target. In such settings, controllers should be designed to limit the worst realizations of the control cost rather than its mean \cite{kishida2023,zhang2024cvar,liu2021}. There is therefore a need for measurement strategies that are optimal with respect to design measures other than the mean. 

Such criteria are formalized by the \emph{optimized certainty equivalent} (OCE) \cite{bental}, a family of objectives that contains the average utility as a special case. Its most widely used member is the \emph{conditional value at risk} (CVaR) \cite{rockafellar}, which averages the utility over its worst outcomes. Risk-averse criteria have been employed in quantum systems for feedback control \cite{james2004} and for filtering \cite{yamamoto2009} using an exponential utility, and quantum algorithms have been proposed to estimate risk measures \cite{woerner, egger2021}. However, the design of measurements that maximize a risk-averse, OCE, criterion has not been addressed. 

In this letter, we study the design of quantum measurements that maximize the OCE utility. We show that, for a piecewise linear gain function, defining the OCE, thus including the CVaR, the problem reduces to a finite number of SDPs, and we derive the corresponding dual formulation. For the discrimination of two states, we obtain a closed-form solution in the form of a Helstrom measurement \cite{helstrom}. Numerical results quantify the resulting reshaping of the utility distribution. Sec. \ref{sec:setting} introduces the setting, Sec. \ref{sec:avg} considers the average utility, Sec. \ref{sec:oce} addresses the OCE utility, and Sec. \ref{sec:num} presents numerical results.

\vspace{-2mm}
\section{Setting}\label{sec:setting}
We study risk-averse decision making for a quantum system. A system in an unknown state is measured, and the outcome is taken as the action. The true state and the action determine a utility, and the measurement is designed to maximize its OCE.

\vspace{-1mm}
\subsection{OCE Utility}
Let the gain function $f:\mathbb{R}\to\mathbb{R}$ be closed, concave and nondecreasing, with $f(0)=0$ and $1\in \partial f(0)$, where $\partial f(0)$ is the superdifferential of $f$ at the origin, which reduces to $f'(0)=1$ when $f$ is differentiable at $0$.
For a utility random variable $u$, the OCE utility is defined as \cite{bental}
\begin{equation}\label{eq:oce}
    \OCE_f(u)=\sup_{t\in\mathbb{R}}\big\{t+\mathbb{E}\big[f(u-t)\big]\big\}.
\end{equation}
Definition \eqref{eq:oce} can be interpreted as the sum of a deterministic reserve $t$ and of the expected gain $\mathbb{E}[f(u-t)]$ obtained from the random residual $u-t$; the OCE selects the best such decomposition by maximizing over the reserve $t$.

Notable examples of OCE measures include the average utility $\OCE_f(u)\hspace{-0.5mm}=\hspace{-0.5mm}\mathbb{E}[u]$, obtained with $f(u)\hspace{-0.5mm}=\hspace{-0.5mm}u$, and the CVaR at level $\alpha\hspace{-0.5mm}\in\hspace{-0.5mm}(0,1]$, obtained with $f_{\alpha}(u) = \hspace{-1mm}\alpha^{-1}\min\{u,0\}$ \cite{rockafellar}. The CVaR, denoted $\CVaR_{\alpha}(u)=\OCE_{f_{\alpha}}(u)$, equals the average of the utility over the worst $\alpha$-fraction of cases.

Generalizing both examples, we take $f(u)$ to be continuous piecewise linear in $u$, with $K$ kinks at arbitrary values $u_1<u_2<\dots<u_K$. The kinks split the domain into $K+1$ segments, and we let $c_k$ denote the slope of $f$ on the $k$-th segment. Concavity and monotonicity of $f$ then amount to
\begin{equation}\label{eq:slopes}
    c_0\geq c_1\geq\dots\geq c_K\geq0,
\end{equation}
and the normalization $1\in\partial f(0)$ implies
\begin{equation}\label{eq:norm}
    c_K\leq1\leq c_0.
\end{equation}
The average utility corresponds to $K=0$ and $c_0=1$, and the CVaR to $K=1$ with $u_1=0$, $c_0=\alpha^{-1}$, and $c_1=0$.

By selecting a sufficiently large $K$, this class of functions can approximate uniformly, on any bounded interval, an arbitrary gain function $f$, such as $f(u)=\lambda^{-1}(1-e^{-\lambda u})$, which gives the entropic risk $\OCE_f(u)=-\lambda^{-1}\log\mathbb{E}[e^{-\lambda u}]$, with risk aversion increasing in $\lambda>0$.

\vspace{-3mm}
\subsection{Problem Definition}
Fix a quantum ensemble $\mathcal{E}=\{p(s),\rho_s\}_{s\in\Sset}$ on a $d$-dimensional Hilbert space, so that the state is given by density matrix $\rho_s$ with probability $p(s)$, where the index $s$ ranges over a finite set $\Sset$. Since the state is not observed directly, the decision maker acquires information about it by measuring the system with a positive operator-valued measure (POVM) $M=\{M_a\}_{a\in\Aset}$, with positive semidefinite matrices $M_a\succeq 0$ satisfying $\sum_{a\in\Aset} M_a = I$ over a finite set $\Aset$, and we let $\Mset$ denote the set of all such POVMs.

By the Born rule \cite{helstrom}, the outcome of the measurement is a random variable $a$ with conditional distribution $p(a| s)=\Tr(\rho_sM_a)$. The quality of the outcome is measured against the true state $s$ through a non-negative utility $u(s,a)$. For example, in state discrimination, the goal is to identify the state, and one typically sets $\Aset = \Sset$ with utility $u(s,a)=\mathds{1}\{s = a\}$ \cite{barnett}.

A measurement $M$ induces the joint distribution
\begin{equation}\label{eq:joint}
    p_M(s,a) = p(s) \Tr(\rho_s M_a)
\end{equation}
over the state index $s$ and the action $a$, and hence over the distribution of the utility $u(s,a)$. We are interested in the problem of optimizing the measurement $M$ with the goal of maximizing the OCE utility
\begin{equation}\label{eq:problem}
    \max_{M\in\Mset} \OCE_f \big(u(s,a)\big), \quad (s,a) \sim p_M(s,a),
\end{equation}
for a given gain function $f$. Since set $\Mset$ is compact and the objective is continuous in $M$, the maximum in \eqref{eq:problem} is attained.

\begin{remark}[Measurement outcomes as actions] \label{rem:fold}
The formulation \eqref{eq:problem} is equivalent to a two-stage setting that may seem more general. In this setting, the decision maker first applies a POVM $\{N_o\}_{o\in\Oset}$, obtaining an observation $o$ with conditional distribution $p(o| s)=\Tr(\rho_sN_o)$, and then selects an action $a$ through a classical policy $\pi(a| o)$.
Measuring the operators $M_a \hspace{-0.8mm}=\hspace{-0.8mm} \sum_{o\in\Oset}\hspace{-0.8mm}\pi(a| o)N_o$, which satisfy $M_a \hspace{-0.8mm}\succeq\hspace{-0.8mm} 0$ and $\sum_{a\in\Aset}M_a \hspace{-0.8mm}=\hspace{-0.8mm} \sum_{o\in\Oset}N_o \hspace{-0.8mm}=\hspace{-0.8mm} I$,
induces the same joint distribution \eqref{eq:joint}. Since the objective in \eqref{eq:problem} depends on the measurement only through the joint distribution \eqref{eq:joint}, there is no loss of optimality in restricting the design to single-stage POVMs.
\end{remark}

\begin{remark}[Commuting states]\label{rem:commuting}
Consider the classical case in which all states $\rho_s$ in the ensemble $\mathcal{E}$ commute. Let $\{|o\rangle\}_{o\in[d]}$ be a common eigenbasis of all states, where $[d]=\{1,\dots,d\}$. Since $\rho_s$ is positive semidefinite with unit trace, its eigenvalues $p(o| s)=\langle o|\rho_s|o\rangle$ form a conditional distribution over $[d]$ for each $s$. For an arbitrary POVM $M$, the diagonal entries $\pi(a| o)=\langle o|M_a|o\rangle$ of $M_a$ in this basis form a conditional distribution over $\Aset$ for each $o$, and every such distribution is realized by the POVM $M_a=\sum_{o\in[d]}\pi(a| o)|o\rangle\langle o|$. Using these definitions, the Born rule becomes
\begin{equation}\label{eq:classicalstep}
    \Tr(\rho_sM_a) = \sum_{o\in[d]}p(o| s)\pi(a| o),
\end{equation}
and thus the objective in \eqref{eq:problem} depends on the measurement $M$ only through the conditional distribution $\pi(a| o)$. It follows that in this special setting it is optimal to first carry out the projective measurement $\{|o\rangle\langle o|\}_{o\in[d]}$ in the common eigenbasis, and then to apply a classical action policy $\pi(a| o)$, to be optimized, that maps observations $o$ into actions $a$.
\end{remark}

\vspace{-3mm}
\section{Average Utility Maximization} \label{sec:avg}
Consider first the conventional case of the average utility, recovered with the gain function $f(u)=u$, i.e., with $K=0$, so that problem \eqref{eq:problem} becomes
\begin{align}
    \max_{M\in\Mset} \mathbb{E}\big[u(s,a)\big] & =\max_{M\in\Mset} \sum_{s\in\Sset} p(s)\sum_{a\in\Aset}\Tr(\rho_sM_a) u(s,a)\nonumber\\
    & =\max_{M\in\Mset} \sum_{a\in\Aset}\Tr\big(U_aM_a\big),\label{eq:avg}
\end{align}
where we have introduced the matrix
\begin{equation}\label{eq:Ua}
    U_a=\sum_{s\in\Sset}p(s)u(s,a)\rho_s.
\end{equation}
The matrix \eqref{eq:Ua}, which is referred to as the utility observable for action $a$, is positive semidefinite, i.e., $U_a\succeq0$, and summarizes all relevant knowledge about the ensemble $\mathcal{E}$.

\begin{remark}[Commuting states (continued)]
For commuting states, substituting \eqref{eq:classicalstep} into \eqref{eq:avg} gives the maximization
\begin{equation}\label{eq:classical-1}
    \max_{\pi(a| o)}\sum_{o\in[d]}p(o) \sum_{a\in\Aset}\pi(a| o) \sum_{s\in\Sset} p(s| o)u(s,a),
\end{equation}
where the optimization is now over the action policy $\pi(a| o)$, with $p(o)=\sum_{s\in\Sset}p(s)p(o| s)$ the marginal distribution of the observation and $p(s| o)=p(s)p(o| s)/p(o)$ the corresponding posterior. The solution is given by the deterministic policy returning, for each observation $o$, any action $a^{\star} \in \arg\max_{a\in\Aset} \sum_{s\in\Sset} p(s| o)u(s,a)$.
\end{remark}

Problem \eqref{eq:avg} can be written more explicitly as
\begin{align}
    \max_{\{M_a\}} & \sum_{a\in\Aset}\Tr\big(U_a M_a\big)\nonumber\\
    \text{s.t.}~ & M_a\succeq 0 \quad \forall a\in\Aset,\nonumber\\
    & \sum_{a\in\Aset}M_a=I. \label{eq:primal}
\end{align}
This is an SDP, which coincides with the problem of quantum statistical decision theory \cite{helstrom,holevo,ykl,barnett} and can be solved using standard tools. By strong duality, which holds since the primal \eqref{eq:primal} admits the strictly feasible point $M_a=I/|\Aset|$, the maximum average utility is also given by the dual problem
\begin{equation}\label{eq:dual}
    \min_{\Lambda \succeq 0} \Tr\Lambda\quad \text{s.t.} ~\Lambda\succeq U_a\quad\forall a\in\Aset.
\end{equation}
The attainable average utility is the smallest trace of a positive semidefinite operator dominating all utility observables \cite{coutts2021certifying}.

\vspace{-1mm}
\section{OCE Utility Maximization} \label{sec:oce}
We now extend the design to a general OCE utility. Substituting \eqref{eq:oce} into \eqref{eq:problem} and using \eqref{eq:joint} gives the problem
\begin{align}
    & \max_{M\in\Mset} \sup_{t\in\mathbb{R}}\bigg\{t + \sum_{s\in\Sset} p(s) \sum_{a\in\Aset}\Tr(\rho_s M_a)f\big(u(s,a)-t\big)\bigg\}\nonumber\\
    & \quad=\sup_{t\in\mathbb{R}}\bigg\{g(t)=t+\max_{M\in\Mset}\sum_{a\in\Aset}\Tr\big(U_a(t) M_a\big)\bigg\}, \label{eq:oceprob}
\end{align}
where the two maximizations have been exchanged in the equality, and we have introduced the tilted utility observables
\begin{equation}\label{eq:tilted}
    U_a(t)=\sum_{s\in\Sset}p(s)f\big(u(s,a)-t\big)\rho_s.
\end{equation}
Unlike the utility observables \eqref{eq:Ua}, the tilted observables \eqref{eq:tilted} are Hermitian but generally not positive semidefinite, since the gain function $f$ may take negative values. The next proposition shows that problem \eqref{eq:oceprob} can be solved by addressing at most $K\,|\Sset|\,|\Aset|$ SDPs akin to \eqref{eq:primal}, where we recall that $K$ is the number of kinks in the piecewise linear function $f$.

\begin{proposition}\label{prop:grid}
Let $f(u)$ be continuous and piecewise linear with $K\geq1$ kinks at $u_1<u_2<\dots<u_K$ and slopes satisfying \eqref{eq:slopes} and \eqref{eq:norm}, and define the finite set
\begin{equation}\label{eq:grid}
    \Tset=\big\{u(s,a)-u_k:s\in\Sset,a\in\Aset,k\in[K]\big\}
\end{equation}
of cardinality at most $K\,|\Sset|\,|\Aset|$. Then the supremum in \eqref{eq:oceprob} is attained at some $\tau\in\Tset$, and an optimal solution of \eqref{eq:oceprob} is obtained by solving the SDP
\begin{align}
    \max_{\{M_a\}} & \sum_{a\in\Aset}\Tr\big(U_a(\tau)M_a\big)\nonumber\\
    \text{\rm{s.t.}}~ & M_a\succeq0 \quad\forall a\in\Aset,\nonumber\\
    & \sum_{a\in\Aset}M_a=I \label{eq:primal-1}
\end{align}
for all $\tau\in\Tset$, and then selecting the pair $(\tau,M)$ that maximizes the objective $\tau + \sum_{a \in \Aset}\Tr(U_a(\tau) M_a)$. For $K=0$, the OCE utility reduces to the average utility and the single SDP \eqref{eq:primal} suffices.
\end{proposition}

\textit{Proof:} See Appendix~\ref{apx_prop}.

The following result is a direct consequence of the dual formulation \eqref{eq:dual} and of Proposition \ref{prop:grid}.

\begin{corollary}[Dual formulation]\label{cor:dual}
The optimal value of problem \eqref{eq:oceprob} equals $\max_{\tau\in\Tset} \{\tau + \Tr(\Lambda^\star(\tau))\}$, where $\Lambda^\star(\tau)$ is the optimal solution of the dual problem
\begin{equation}\label{eq:dual_prop}
    \min_{\Lambda=\Lambda^{\dagger}} \Tr\Lambda \quad \text{\rm{s.t.}} ~ \Lambda\succeq U_a(\tau) \quad \forall a\in\Aset.
\end{equation}
\end{corollary}

The next result elaborates on the special case of binary state discrimination, which yields a specific type of Helstrom measurement \cite{helstrom}.

\begin{corollary}[Binary state discrimination] \label{cor:binary}
Consider the binary case $\Sset=\Aset=\{0,1\}$, with a utility satisfying $u(s,s)\geq u(s,1-s)$, so that a correct decision is never worse than an incorrect one. For $\tau\in\Tset$, define the weights
\begin{equation}\label{eq:weights}
    w_s(\tau) = p(s)\big[f\big(u(s, s)-\tau\big) - f\big(u(s, 1-s)-\tau\big)\big] \geq 0.
\end{equation}
Then an optimal solution of problem \eqref{eq:oceprob} is given by
\begin{equation}\label{eq:binary-M}
    M_0=\Pi_+\big(\eta(\tau^\star)\rho_0-(1-\eta(\tau^\star)) \rho_1\big), \quad M_1=I-M_0,
\end{equation}
with $\eta(\tau) = w_0(\tau) / \big(w_0(\tau) + w_1(\tau)\big)$, where $\Pi_+(A)$ is the projector onto the eigenvectors of $A$ with positive eigenvalues, and the value $\tau^\star\in\Tset$ in \eqref{eq:binary-M} is selected as
\begin{align}\label{eq:binary-g}
    \tau^\star=&\arg\max_{\tau\in\Tset}\Big\{\tau+\Tr(U_1(\tau)) \nonumber\\
    &+ \hspace{-0.5mm} \frac{1}{2}\Big[w_0(\tau)\hspace{-0.5mm} - \hspace{-0.5mm} w_1(\tau) \hspace{-0.5mm}+ \hspace{-0.5mm} \big\|w_0(\tau)\rho_0 \hspace{-0.5mm}-\hspace{-0.5mm} w_1(\tau)\rho_1\big\|_1\Big]\Big\},
\end{align}
with $\|\cdot\|_1$ being the trace norm.
\end{corollary}

\textit{Proof:} See Appendix \ref{apdx_cor_2}.

For the average utility, with $f(u)\hspace{-0.8mm}=\hspace{-0.8mm}u$ and the detection utility $u(s,a)\hspace{-0.5mm}=\hspace{-0.5mm}\mathds{1}\{s=a\}$, the measurement \eqref{eq:binary-M} recovers the standard Helstrom measurement $\Pi_+\hspace{-0.5mm} \big(p(0)\rho_0\hspace{-0.5mm}-\hspace{-0.5mm}p(1)\rho_1\big)$.

\vspace{-1mm}
\section{Numerical Results} \label{sec:num}
We consider a Hilbert space of dimension $d=8$, sets
$\Sset=\Aset=\{1,2,3,4\}$, and the non-uniform prior $p(s)=(0.15, 0.2, 0.25, 0.4)$. The states are given by
\begin{equation}\label{eq:states}
    \rho_s=\gamma\sigma_s+(1-\gamma)C(\sigma_s),\quad\gamma\in[0,1],
\end{equation}
where $\sigma_s=G_sG_s^{\dagger}/\Tr(G_sG_s^{\dagger})$ is a random density operator of rank $r=4$, the entries of $G_s\in\mathbb{C}^{d\times r}$ are sampled i.i.d. from the standard complex Gaussian distribution $\mathcal{C N}(0,1)$, and $C(\sigma)=\sum_{o\in[d]}|o\rangle\langle o|\langle o|\sigma|o\rangle$ is the dephasing channel in the computational basis $\{|o\rangle\}_{o\in[d]}$. The parameter $\gamma$ controls the coherence of state $\rho_s$ with respect to the computational basis \cite{simeone2026}, with $\gamma=0$ returning incoherent, diagonal states.

The utility is $u(s,a)=1-|a-s|/(|\Sset|-1)$,
which decreases with the distance between the true state index $s$ and the action $a$. We consider the average utility, obtained with $f(u)=u$, the CVaR at level $\alpha=0.2$, and a more complex OCE utility described below. As a benchmark, we restrict the POVM to be diagonal in the computational basis, $M_a=\sum_{o\in[d]}\pi(a| o)|o\rangle\langle o|$, and optimize over the classical policy $\pi(a| o)$. By Remark~\ref{rem:commuting}, this two-stage approach is optimal at $\gamma=0$. The SDPs \eqref{eq:primal-1} are solved with CVX \cite{grant2014cvx}.

\subsubsection{On the Optimal Solution}
We first demonstrate the key result in Proposition \ref{prop:grid} that the supremum over variable $t$ in \eqref{eq:oceprob} is attained on the finite set $\Tset$ in \eqref{eq:grid}. We choose an OCE utility with $K=2$ kinks at $u_1=0$, $u_2=1/3$, and slopes $(c_0, c_1, c_2)=(2, 1, 0.2)$, so that set \eqref{eq:grid} is given by $\Tset=\{\tau_1=-1/3, \tau_2=0, \tau_3=1/3, \tau_4=2/3, \tau_5=1\}$.
Fig. \ref{fig_prop} plots function $g(t)$ in \eqref{eq:oceprob} versus $t$. Within each interval $[\tau_{j-1}, \tau_j]$, the function $g(t)$ is seen to be convex, and outside the interval $[\tau_1,\tau_5]$, the function $g(t)$ is nondecreasing to the left and nonincreasing to the right. The maximum of function $g(t)$ is therefore attained on one of the points in set $\Tset$.

\begin{figure}[t]
	\centering
	\includegraphics[width = 0.65\linewidth]{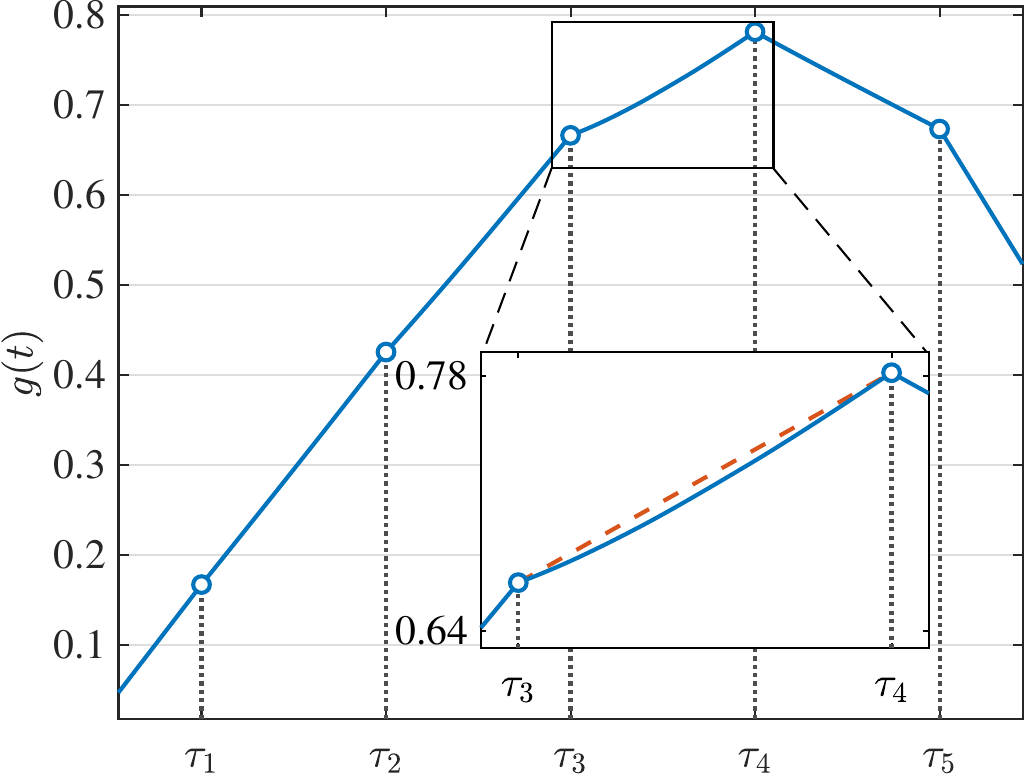}
	\caption{An example of objective function $g(t)$ in \eqref{eq:oceprob} versus $t$.}\label{fig_prop}
    \vspace{-2mm}
\end{figure}

\begin{figure}[t]
	\centering
	\includegraphics[width = 0.65\linewidth]{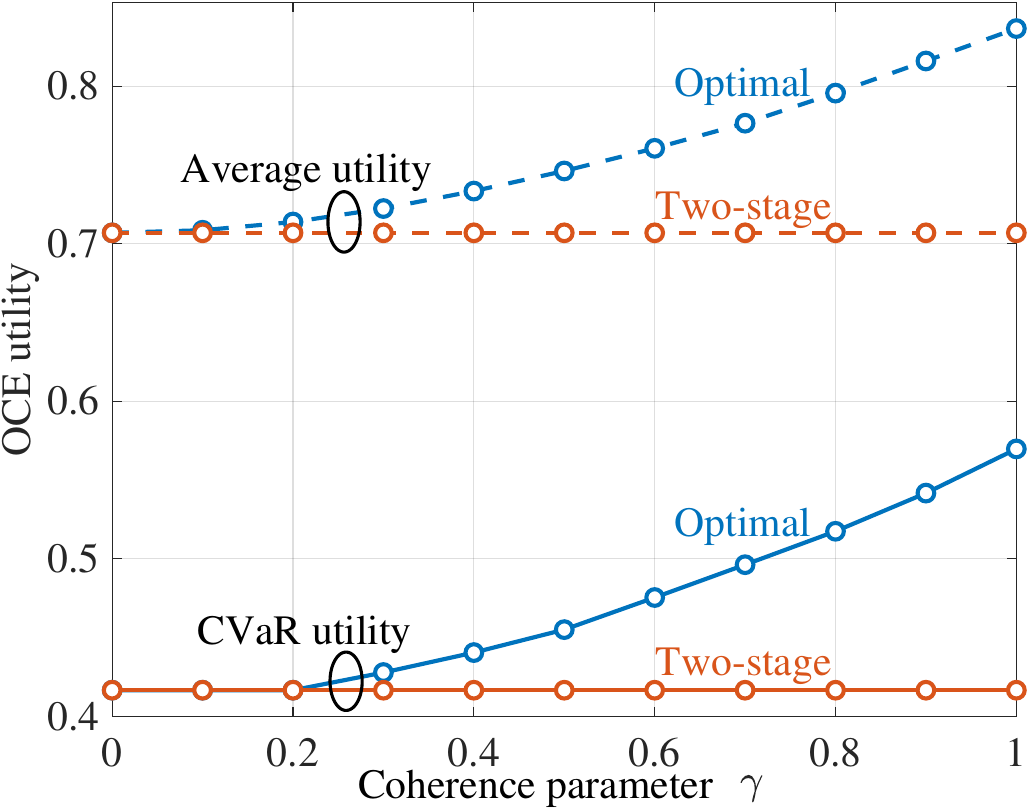}
	\caption{Average utility and CVaR utility versus the coherence parameter $\gamma$ for the optimal measurement and the benchmark two-stage process with standard measurement in the computational basis ($\alpha=0.2$).}\label{fig_measure}
    \vspace{-2mm}
\end{figure}

\begin{figure}[t]
	\centering
	\includegraphics[width = 0.65\linewidth]{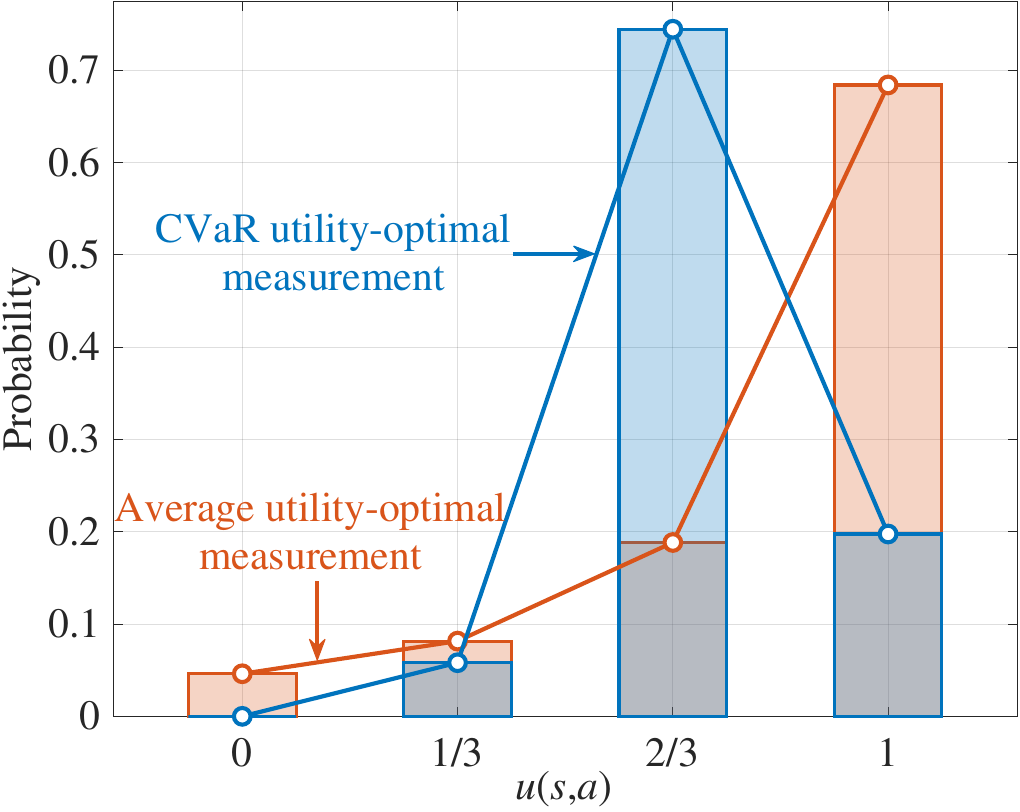}
	\caption{Empirical distribution of the utility $u(s,a)$, obtained by the optimal measurement under the average utility and the CVaR ($\alpha=0.2$, $\gamma=1$).} \label{fig_util}
    \vspace{-2mm}
\end{figure}

\subsubsection{On the Measurement Design}
We now compare the optimal POVM with the two-stage benchmark based on computational-basis measurements. Fig. \ref{fig_measure} plots the average utility and the CVaR versus the coherence parameter $\gamma$ in \eqref{eq:states}.
The utility attained by the two-stage process does not depend on the coherence parameter $\gamma$, since by \eqref{eq:classicalstep} it depends on the states only through the diagonal entries $\langle o|\rho_s|o\rangle = \langle o|\sigma_s|o\rangle$. By Remark \ref{rem:commuting}, the two-stage procedure and the optimal measurement have the same performance for $\gamma=0$. As $\gamma$ increases, however, the gap between the two measurements increases for both average utility and CVaR.

\subsubsection{On the Effect of Risk Aversion}
We finally compare the empirical distribution of the utility $u(s,a)$ induced by the optimal measurement under the average utility and the CVaR for $\gamma=1$. As shown in Fig. \ref{fig_util}, compared with the average utility, the CVaR improves the distribution at the lower tail, making the worst outcomes much less likely at the cost of a lower average utility.

\vspace{-3mm}
\section{Conclusions}
In this letter, we have studied the design of quantum measurements that maximize an OCE utility. For a piecewise linear gain function, e.g., for the CVaR, the design reduces to a finite number of SDPs, in place of a single SDP for the average utility, with a closed-form solution for the discrimination of two states. Numerical results show that risk aversion makes the worst outcomes much less likely at a moderate cost in average utility, and that the gain over a fixed computational-basis measurement grows with the coherence of the states.

\vspace{-3mm}
\appendix

\vspace{-1mm}
\subsection{Proof of Proposition \ref{prop:grid}}\label{apx_prop}
For fixed $s$ and $a$, the function $f\big(u(s,a)-t\big)$ changes slope only when $u(s,a)-t$ crosses a kink, that is, at the values of $t$ in $\Tset(s,a)\hspace{-0.5mm}=\hspace{-0.5mm}\{u(s,a)-u_k:k\in[K]\}$, and is affine between consecutive such values. Let $\tau_1<\dots<\tau_J$ be the ordered entries of $\Tset$ in \eqref{eq:grid}, which collects the values in $\Tset(s,a)$ over all $s$ and $a$. On each interval $[\tau_j,\tau_{j+1}]$, the objective $t\hspace{-0.1mm}+\hspace{-0.2mm}\sum_{a\in\Aset}\hspace{-0.1mm}\Tr\big(U_a(t) M_a\big)$ in \eqref{eq:oceprob} is affine in $t$ for every $M\hspace{-0.5mm}\in\hspace{-0.5mm}\Mset$, so that $g(t)$ is a pointwise maximum of affine functions and is convex there. Its maximum over the interval is attained at one of the two endpoints, both of which belong to $\Tset$.

It remains to exclude the two unbounded intervals. For $t\geq\tau_J$, we have $u(s,a)-t\leq u_1$ for all $s$ and $a$, and hence $f\big(u(s,a)-t\big)=c_0\big(u(s,a)-t\big)+b_0$ for a suitable constant $b_0$. Using $c_0\geq0$ from \eqref{eq:slopes}, this yields $g(t)=(1-c_0)t + c_0 \max_{M\in\Mset}\mathbb{E}\big[u(s,a)\big] + b_0$, which is nonincreasing in $t$ by \eqref{eq:norm}. Symmetrically, for $t\leq\tau_1$ we have $u(s,a)-t\geq u_K$ for all $s$ and $a$, so that $g(t)=(1-c_K)t+c_K\max_{M\in\Mset}\mathbb{E}[u(s,a)]+b_K$ for a suitable constant $b_K$, which is nondecreasing in $t$ by \eqref{eq:norm}. The supremum of $g(t)$ is therefore attained on $[\tau_1,\tau_J]$, and thus at one of the values in $\Tset$.

\vspace{-1mm}
\subsection{Proof of Corollary \ref{cor:binary}} \label{apdx_cor_2}
Substituting $M_1=I-M_0$ into the objective of \eqref{eq:primal-1} gives
\begin{equation}\label{eq:binary-split}
    \sum_{a\in\Aset}\Tr\big(U_a(\tau)M_a\big)=\Tr \big(U_1(\tau)\big) + \Tr\big(A(\tau)M_0\big),
\end{equation}
where only the second term depends on $M_0$ and we have introduced the matrix
\begin{equation}\label{eq:binary-A}
    A(\tau) = U_0(\tau)-U_1(\tau)=w_0(\tau)\rho_0-w_1(\tau)\rho_1,
\end{equation}
with the second equality following from \eqref{eq:tilted} and \eqref{eq:weights}. Since \eqref{eq:binary-A} is Hermitian, its spectral decomposition $A(\tau) = \sum_{i\in[d]}\lambda_i|i\rangle\langle i|$ gives $\Tr\big(A(\tau)M_0\big) = \sum_{i\in[d]}\lambda_i \langle i|M_0|i\rangle$,
where $0\preceq M_0\preceq I$ gives $0\leq\langle i|M_0|i \rangle\leq1$. The sum
is therefore maximized by $M_0=\Pi_+\big(A(\tau)\big)$, which gives \eqref{eq:binary-M} after normalizing $A(\tau)$ by $w_0(\tau)+w_1(\tau)$. 


\clearpage
\bibliographystyle{IEEEtran}
\bibliography{refs}

\end{document}